%% file: walter_lehmann_algorithmic_differentiation_of_linear_algebra_functions.tex
\RequirePackage{ifpdf}
\ifpdf
\documentclass[a4paper,12pt,twoside]{article}
\usepackage{microtype}
\else
\documentclass[a4paper,12pt,twoside,dvips]{article}
\fi

\usepackage[scale=0.85,includeheadfoot]{geometry}
\usepackage[medium]{titlesec}

\usepackage[utf8x]{inputenc}
\usepackage[english]{babel}

\usepackage[T1]{fontenc}

\usepackage{amsmath,amsfonts,amsthm, mathtools}
\usepackage{subeqnarray}
\usepackage{latexsym}

\usepackage{mathptmx}

\usepackage{parskip}
\usepackage{fancyhdr}
\usepackage{graphicx}
\usepackage{hyperref}

\usepackage{tikz}

\graphicspath{{pics/}}

\input{walter_lehmann_document_header}

\title{Algorithmic Differentiation of Linear Algebra Functions with Application in Optimum Experimental Design (Extended Version)}

\author{Sebastian~F. Walter\footnote{\texttt{sebastian.walter@gmail.com}}
 \and Lutz Lehmann\footnote{\texttt{llehmann@mathematik.hu-berlin.de}}}

\begin{document}
\maketitle

\begin{abstract}
We derive algorithms for higher order derivative computation of the rectangular $QR$ and eigenvalue decomposition of symmetric matrices with distinct eigenvalues in the forward and reverse mode of algorithmic differentiation (AD) using univariate Taylor propagation of matrices (UTPM). Linear algebra functions are regarded as elementary functions and not as algorithms.
The presented algorithms are implemented in the BSD licensed AD tool \texttt{ALGOPY}. Numerical tests show that the UTPM algorithms derived in this paper produce results close to machine precision accuracy. The theory developed in this paper is applied to compute the gradient of an objective function motivated from optimum experimental design: $\nabla_x \Phi(C(J(F(x,y))))$, where $\Phi = \{\lambda_1 : \lambda_1 \mbox{largest eigenvalue of }C\}$, $C = (J^T J)^{-1}$, $J = \frac{\dd F}{\dd y}$ and $F = F(x,y)$.
\end{abstract}

\tableofcontents

\input{walter_lehmann_document.tex}

\bibliographystyle{plain}
\bibliography{refs}

\end{document}

%% file: walter_lehmann_document_header.tex
\newcommand{\R}{{\mathbb R}}
\newcommand{\K}{{\mathbb K}}

\newcommand{\dd}{{\rm d}}
\newcommand{\Id}{  1 \!\!\! \mathrm I}
\newcommand{\solve}{{\rm solve \;}}
\newcommand{\qr}{{\rm qr \;}}
\newcommand{\eig}{{\rm eig \;}}
\newcommand{\tr}{{\rm tr \;}}

\newcommand{\eqE}{ \stackrel{E}{=}}
\newcommand{\eqD}{ \stackrel{D}{=}}
\newcommand{\eqDE}{ \stackrel{D+E}{=}}
\newcommand{\pfe}[1]{ \lbrack #1 \rbrack_E}
\newcommand{\pfd}[1]{ \lbrack #1 \rbrack_D}
\newcommand{\pfde}[1]{ \lbrack #1 \rbrack_{D+E}}

\newcommand{\textred}[1]{#1}
\newcommand{\textblue}[1]{#1}
\newcommand{\textgreen}[1]{#1}

\newcommand{\pb}[1]{ \overleftarrow{P}{(#1)} }
\newcommand{\pf}[1]{ \overrightarrow{P}{(#1)} }

\newcommand{\add}{\mathrm{add}\;}

\newcommand{\mul}{\mathrm{mul}\;}

\newcommand{\Eqn}[1]{Eqn. (\ref{#1})}

\newcommand{\QDE}{\textred{[Q]_{D+E}}}
\newcommand{\QTDE}{\textred{[Q^T]_{D+E}}}
\newcommand{\QD}{\textred{[Q]_D}}
\newcommand{\QTD}{\textred{[Q^T]_D}}
\newcommand{\QE}{\textred{[Q]_E}}
\newcommand{\QTE}{\textred{[Q^T]_E}}

\newcommand{\LDE}{\textred{[\Lambda]_{D+E}}}
\newcommand{\LE}{\textred{[\Lambda]_E}}
\newcommand{\LD}{\textred{[\Lambda]_D}}

\newcommand{\ADE}{\textblue{[A]_{D+E}}}
\newcommand{\AD}{\textblue{[A]_D}}

\newcommand{\DAE}{\textblue{[\Delta A]_E}}
\newcommand{\DFE}{[\Delta F]_E}
\newcommand{\DGE}{[\Delta G]_E}
\newcommand{\DLE}{\textgreen{[\Delta \Lambda]_E}}
\newcommand{\DQE}{\textgreen{[\Delta Q]_E}}

\theoremstyle{plain}
\newtheorem{proposition}{Proposition}

\newtheorem{lemma}[proposition]{Lemma}

\theoremstyle{definition}

\newtheorem{algorithm}[proposition]{Algorithm}

\theoremstyle{remark}

%% file: walter_lehmann_document.tex
\section{Introduction}
\label{intro}
The theory of \emph{Algorithmic Differentiation} (AD) is concerned with the automated generation of efficient algorithms for derivative computation of computational models. A \emph{computational model} (CM) is the description of a mathematically expressed (scientific) problem as a computer program. That means that the CM is a composite function of elementary functions. From a mathematical point of view, only the operations $*,+$ together with their inverse operations $/,-$ are elementary functions since they are required to define the field of real numbers $\R$. However, there are good reasons to include other functions, e.g. those defined in the C-header \texttt{math.h}. The reason is firstly because algorithmic implementations of functions as \texttt{exp,sin} may contain non-differentiable computations and branches, furthermore one can use the additional structure to derive more efficient algorithms. E.g. to compute the univariate Taylor propagation of $\sin (\sum_{d=0}^{D-1} x_d t^d)$ can be done in $\mathcal O(D^2)$ arithmetic operations by using the structural information that they are solutions of special ordinary differential equations (c.f. \cite{Griewank2008EDP}). Even of higher practical importance is the fact that deriving explicit formulas for functions as the trigonometric functions reduces the memory requirement of the reverse mode of AD to $\mathcal O(D)$ since the intermediate steps do not have to be stored. That means, explicitly deriving derivative formulas for functions can yield much better performance and smaller memory requirements.
In scientific computing, there are many functions that exhibit rich structural information. Among those, the linear algebra functions as they are for example implemented in LAPACK are of central importance. This motivates the authors' efforts to treat linear algebra routines as elementary functions.

\section{Related Work}
Computing derivatives in the forward mode of AD can be done by propagating polynomial factor rings through a program's computational graph. In the past, choices have been univariate Taylor polynomials of the form $[x]_D = \sum_{d=0}^{D-1} x_d t^D$, where $x_d \in \R$ as described in the standard book ``Evaluating Derivatives: Principles and Techniques of Algorithmic Differentiation'' by Griewank \cite{Griewank2008EDP} and implemented e.g. in the AD tool ADOL-C \cite{Griewank1999ACA}. Also, multivariate Taylor polynomials of the form $[x]_D = \sum_{|i| \leq D-1} x_i t^i$, where $i$ is a multi-index and $x_i \in \R$, have been successfully used, e.g. for high-order polynomial approximations \cite{har08,Neidinger1992AEM}.  Univariate Taylor polynomials over matrices have also been considered, i.e. $\pfd{A} = \sum_{d=0}^{D-1} A_d t^d $ of fixed degree $D-1$ with matrix valued coefficients $ A_d\in \R^{M\times N}$ for $ d=0,\dots,D-1$. Very close to our work is Eric Phipps' PhD thesis \cite{ETP03}. Phipps used the combined forward and reverse mode of AD for the linear algebra routines $Ax = b, A^{-1}, C = AB, C = A \circ B$ in the context of a Taylor Series integrator for differential algebraic equations. A paper by Vetter in 1973  is also treating matrix differentials and the combination with matrix Taylor expansions \cite{vetter73}. However, the focus on the paper is on results of matrix derivatives of the form $D_B A(B) \in \R^{M_A M_B \times N_A N_B}$, where $A \in \R^{M_A, N_A}, B \in \R^{M_B, N_B}$ and the derivative computation is not put into the context of a computational procedure.

An alternative to the Taylor propagation approach is to transform the computational graph to obtain a computational procedure that computes derivatives. Higher order derivatives are computed by successively applying these transformations. There is a comprehensive and concise reference for first order matrix derivatives in the forward and reverse mode of AD by Giles \cite{Giles2008CMD}. More sophisticated differentiated linear algebra algorithms are collected in an extended preprint version \cite{Giles2007CMDextended}. The eigenvalue decomposition algorithm derived by Giles is a special case of the algorithm presented in this paper.

\section{Mathematical Description}
\subsection{Notation for Composite Functions}
Typically, in the framework of AD one considers functions $F: \R^N \rightarrow \R^M$, $x \mapsto y = F(x)$ that are built of elementary functions $\phi$.
where $x \equiv (x_1,\dots, x_N),\; y \equiv (y_1,\dots, y_M)$ with $x_n, y_m \in \R$ for $1\leq n \leq N, 1 \leq m \leq M$. Instead, we look at functions
\begin{eqnarray}
F : \bigoplus_{n=1}^N \K_n  &\rightarrow& \bigoplus_{m=1}^M \K_m \\
(x_1,\dots,x_N) &\mapsto& (y_1,\dots, y_M) = F(x_1,\dots, x_N) \;,
\end{eqnarray}
where $\K$ is some ring. Here $\K = (\R^{M \times N}, +, \cdot)$ where $+,\cdot$ the usual matrix matrix addition and multiplication..
If $F$ maps to $\R^{1 \times 1}$ we use the symbol $f$ instead of $F$.
For example $f(x,y) = \tr ( xy + x)$ where $x \in \R^{N, N}, y \in \R^{N,N}$ can be written as
\begin{equation*}
f(x,y) = \phi_4( \phi_3( \phi_1(x,y), \phi_2(x)))  = \phi_4 (\phi_3(v_1,v_2)) = \phi_4 ( v_3).
\end{equation*}
We use the notation $v_l$ for the result of $\phi_l$ and $v_{j \prec l}$ for all arguments of $\phi_l$. To be consistent the independent input arguments $v_n$ are also written as $v_{n-N} = x_n$. To sum it up, the following three equations describe the function evaluation:
\begin{eqnarray}
v_{n-N} & = & x_n \quad\quad\quad n=1,\dots,N\\
v_{l} & = & \phi_l(v_{j \prec l}) \quad l =1,\dots, L\\
y_{M-m} & = & v_{L-m} \quad\quad m = M-1, \dots, 0 \;,
\end{eqnarray}
where $L$ is the number of calls to basics functions $\phi_l$ during the computation of $F$.

\subsection{The Push Forward}
We want to lift the computational procedure to work on the polynomial factor ring $ \K[t] / t^D \K[t]$ with representatives $[x]_D := \sum_{d=0}^{D-1} x_d t^d$.
We define the \emph{push forward} of a sufficiently smooth function in the following way:
\begin{eqnarray}
\pf{f} ([x]_D) &:=& \sum_{d=0}^{D-1} \frac{1}{d!} \left. \frac{\dd^d}{\dd t^d} f (   \sum_{c=0}^{D-1} x_c t^c) \right|_{t=0} t^d \;.
\end{eqnarray}
For a function $y = f(x)$ we use the notation $ [y]_D = \pf{f}([x]_D)$.
The definition of the push forward induces the usual addition and multiplication of ring elements
\begin{eqnarray}
\pf{\mul}([x]_D, [y]_D) &=&  \sum_{d=0}^{D-1} \frac{1}{d!} \left.  \frac{\dd^d}{\dd t^d}   (\sum_{c=0}^{D-1} x_c t^c)(\sum_{k=0}^{D-1} y_k t^k) \right|_{t=0} t^d \;,\\
\pf{\add}([x]_D, [y]_D) &=&  \sum_{d=0}^{D-1} \frac{1}{d!} \left.  \frac{\dd^d}{\dd t^d}   (\sum_{c=0}^{D-1} x_c t^c)+(\sum_{k=0}^{D-1} y_k t^k) \right|_{t=0} t^d \;.
\end{eqnarray}

We now look at the properties of this definition:
\begin{proposition}
For $f: \R^K \rightarrow \R^M$ and $g: \R^N \rightarrow \R^K$ sufficiently smooth functions we have 
\begin{eqnarray}
\pf{f \circ g} = \pf{f} \circ \pf{g} \;.
\end{eqnarray}
\end{proposition}
\begin{proof}
Let $[x]_D \in \R^N[t]/t^D \R^N[t]$, then 
\begin{eqnarray*}
\pf{f \circ g} ([x]_D) &=&  \sum_{d=0}^{D-1} \frac{1}{d!} \left. \frac{\dd^d}{\dd t^d} (f \circ g) (   \sum_{c=0}^{D-1} x_c t^c) \right|_{t=0} t^d \\
&=&  \sum_{d=0}^{D-1} \frac{1}{d!} \left. \frac{\dd^d}{\dd t^d} (f (g (   \sum_{c=0}^{D-1} x_c t^c)) \right|_{t=0} t^d \\
&=&  \sum_{d=0}^{D-1} \frac{1}{d!} \left. \frac{\dd^d}{\dd t^d} (f ( \sum_{c=0}^{D-1} g_c t^c) + \mathcal O(t^D)) \right|_{t=0} t^d \\
&=&  \sum_{d=0}^{D-1} \frac{1}{d!} \left. \frac{\dd^d}{\dd t^d} (f ( \sum_{c=0}^{D-1} g_c t^c)) \right|_{t=0} t^d \\
&=& \left( \pf{f} \circ \pf{g} \right) ( [x]_D) \;.
\end{eqnarray*}
\end{proof}
This proposition is of central importance because it allows us to differentiate any composite function $F$ by providing implementations for the push forwards $P(\phi)$ of a fixed set of elementary functions $\phi$.

\subsection{The Pullback}
The differential $\dd f$ of a function $f: \R^N \rightarrow \R^M, x \mapsto y = f(x)$ is a linear map between the tangent spaces, i.e. $\dd f(x): T_x \R^N \mapsto T_y \R^M$. An element of the cotangent bundle $T^* \R^M$ can be written as
\begin{eqnarray}
\alpha(\bar y,y) := \sum_{m = 1}^M \bar y_m \dd y^m \;,
\end{eqnarray}
i.e. $\alpha(\bar y,y)$ maps any element of $T_y \R^M$ to $\R$. 
\begin{eqnarray}
\alpha( \bar y , y) &=& \sum_{m=1}^M \bar y_m \dd y^m = \sum_{m=1}^M \bar y_m \dd f^m(x) =  \sum_{m=1}^M \bar y_m \sum_{n=1}^N \frac{\partial f^m}{\partial x^n}\dd x^n \\
&=&  \sum_{n=1}^N \sum_{m=1}^M \bar y_m  \frac{\partial y^m}{\partial x^n} \dd x^n = \sum_{n=1}^N \bar x_n \dd x^n = \alpha(\bar x,x) \;.
\end{eqnarray}
A \emph{pullback} of a composite function $f \circ g$ is defined as $\pb{f} := f \circ g$. To keep the notation simple we often use $f \equiv \pb{f}$.

\subsection{The Pullback of Lifted Functions}
We want to lift the function $\alpha$. Due to Proposition \ref{prop:lifted_pullback} we are allowed to decompose the global pullback to a sequence of pullbacks of elementary functions.
\begin{lemma} \label{lemma:pfd_dpf}
\begin{eqnarray}
 \dd (\pf{f}) = \pf{\dd f} \; 
\end{eqnarray}
\end{lemma}
\begin{proof}
Follows from the fact that the differential operators $\frac{\dd}{\dd t}$ and $\dd$ interchange.
\end{proof}
\begin{proposition} \label{prop:lifted_pullback}
Let $v_l = \phi_l(v_{j \prec l})$ and $\phi_l$ some elementary function. $v_{j \prec l}$ are all arguments $v_j$ of $\phi_l$ (c.f. \cite{Griewank2008EDP}) .Then we have
\begin{eqnarray}
\pf{\alpha}( [\bar v_l], [v_l]) &=& \sum_{j \prec l} \pf{\alpha}( [\bar v_j], [v_j]) \;.
\end{eqnarray}
\end{proposition}
\begin{proof}
\begin{eqnarray*}
\pf{\alpha} ([\bar v_l], [v_l]) &=&  [\bar v_l] \dd \pf{ v_l([v_{j \prec l}])} \\
&=& [\bar v_l]_D \pf{ \dd  v_l}([v_{j \prec l}])  \\
&=& \sum_{j \prec l} [\bar v_l]_D \pf{\frac{\partial \phi_l}{\partial v_j}}([v_{j \prec l}]) \pf{\dd [v_j]} \\
&=&  \sum_{j \prec l} [\bar v_j] \dd [v_j] \\
&=& \sum_{j \prec l} \pf{\alpha}( [\bar v_j],[v_j])
\end{eqnarray*}
\end{proof}
These propositions tell us that we can use the reverse mode of AD on lifted functions by first computing a push forward and then go reverse step by step where algorithms for $\dd \phi$ must be provided that work on $\R[t]/t \R[t]$. It is necessary to store $\pf{\frac{\partial \phi_l}{\partial v_j}}([v_{j \prec l}])$ or $[v_{j \prec l}]$ during the forward evaluation. From the sum $\sum_{j \prec l}$ one can see that the pullback of a function $\phi_l$ is local, i.e. does not require information of any other $v_i$. In the context of a computational procedure one obtains
\begin{eqnarray}
[\bar v_j] = \sum_{j \prec k} [\bar v_k] \pf{\frac{\partial \phi_k}{\partial v_j}}( [v_{i \prec k}]) \;.
\end{eqnarray}

\subsection{Univariate Taylor Propagation of Matrices}
Now that we have introduced the AD machinery, we look at the two possibilities to differentiate linear algebra (LA) functions. One can regard LA functions as algorithms. Formally, this approach can be written as
\begin{eqnarray}
\left[
\begin{matrix}
[Y_{11}] &\dots& [ Y_{1M_Y}] \\
\vdots & \ddots & \vdots \\
[Y_{N_Y1}] & \dots &[Y_{N_Y M_Y}] \\ 
\end{matrix}
\right] &=&
\pf{F}
\left(
\left[
\begin{matrix}
[X_{11}] &\dots& [X_{1 M_X}] \\
\vdots & \ddots & \vdots \\
[X_{ N_X 1}] & \dots &[X_{N_X M_X}] \\ 
\end{matrix}
\right] \right)\;,
\end{eqnarray}
I.e. the function $F$ is given a matrix with elements $[X_{nm}] \in \R[t]/t \R[t]$.
A simple reformulation transforms such a matrix into a polynomial factor ring over matrices $\K[t]/ t \K[t]$, $\K = \R^{N \times M}$:
\begin{eqnarray}  \label{eqn:utpm}
\left[
\begin{matrix}
\sum_{d=0}^D X_d^{11} t^d &\dots& \sum_{d=0}^D X_d^{1M} t^d \\
\vdots & \ddots & \vdots \\
\sum_{d=0}^D X_d^{N1} t^d & \dots &\sum_{d=0}^D  X_d^{NM} t^d \\ 
\end{matrix}
\right]
&=&
\sum_{d=0}^D
\left[
\begin{matrix}
X_d^{11} & \dots &X_d^{1M}\\
\vdots & \ddots & \vdots \\
X_d^{N1} & \dots & X_d^{NM}\\
\end{matrix}
\right]
t^D \;. 
\end{eqnarray}
We denote from now on matrix polynomials as the rhs of Eqn. (\ref{eqn:utpm}) as $[X]$. The formal procedure then reads
\begin{equation}
[Y] = \pf{F}([X])\;,
\end{equation}
where $\pf{F}$ must be provided as an algorithm on $\K$.

\subsection{Pullback of Matrix Valued Functions}
Applying the reverse mode to a function $f: \R^{N \times M} \rightarrow \R, X \mapsto y = f(x)$ yields
\begin{eqnarray}
\bar y \dd f(X) &=& \sum_{n,m} \bar y \frac{ \partial f}{\partial X_{nm}} \dd X_{nm} \\
&=& \tr \left( 
\underbrace{ \bar f
\left[
\begin{matrix}
 \frac{ \partial f}{\partial X_{11}} & \dots & \frac{ \partial f}{\partial X_{1N}} \\
\vdots & \ddots & \vdots \\
 \frac{ \partial f}{\partial X_{M1}} & \dots & \frac{ \partial f}{\partial X_{MN}} \\
\end{matrix}
\right]
}_{=: \bar X^T \in \R^{M \times N}}
\underbrace{
\left[
\begin{matrix}
 \dd X_{11} & \dots &  \dd X_{1M} \\
\vdots & \ddots & \vdots \\
 \dd X_{N1}& \dots & \dd X_{NM} \\
\end{matrix}
\right]
}_{=: \dd X \in \R^{N \times M}}
\right) \\
&=& \tr ( \bar X^T \dd X ) \;.
\end{eqnarray}
Some well-known results \cite{Giles2007CMDextended,ETP03} of the reverse mode for unary functions $ Y = F(X)$  are
\begin{align} 
Y = X^{-1}     &:& \tr( \bar Y^T \dd Y )    =& \tr ( - Y \bar Y^T Y \dd X )  \label{eqn:matrix_reverse_inverse} \\
Y = X^T        &:& \tr ( \bar Y^T \dd Y )   =& \tr ( \bar Y \dd X ) \\
y = \tr (X) &:& \bar y \dd \tr(X)        =& \tr( \bar y \Id \dd X ) 
\end{align}
For binary functions  $ Z = F(X,Y)$ one obtains $ \tr ( \bar Z^T \dd Z) = \tr ( \bar X^T \dd X) +  \tr ( \bar Y^T \dd Y)$. E.g for the matrix matrix muliplication one has
\begin{align}
Z = X Y        &:& \tr( \bar Z^T \dd Z)       =&   \tr \left( Y \bar Z^T \dd X + \bar Z^T X \dd Y \right) \;.
\end{align}

\section{Preliminaries for the Rectangular $QR$ and Eigenvalue Decomposition}
In this section we establish the notation and derive some basic lemmas that are used in the derivation of the push forward of the rectangular $QR$ and eigenvalue decomposition of symmetric matrices with distinct eigenvalues. Both algorithms have as output special matrices, i.e. the upper tridiagonal matrix $R$ and the diagonal matrix $\Lambda$. We write the algorithms in implicit form for general $\R^{M \times N}$ matrices and enforce their structure by additional equations. E.g. an upper tridiagonal matrix $R$ is an element of $\R^{N \times N}$ satisfying $P_L \circ R = 0$, where the matrix $P_L$ is defined by $(P_L)_{ij} = (i > j)$, i.e. a strictly lower tridiagonal matrix with all ones below the diagonal. The binary operator $\circ$ is the Hadamard product of matrices, i.e. element wise multiplication. We define $ \sum_{d=0}^\infty x_d t^D \eqD \sum_{d=0}^\infty y_d t^D$ iff $x_d = y_d$ for $d=0,\dots,D-1$.
\begin{lemma}\label{lemma_transpose_projector}
Let $A \in \R^{N \times N}$ and $P_L$ resp. $P_R$ defined as above. Then
\begin{eqnarray}
(P_L \circ A)^T &=& P_R \circ A^T \;.
\end{eqnarray}
\end{lemma}
\begin{proof}
\begin{eqnarray*}
B_{ij} &:=& (P_L \circ A)_{ij} = A_{ij} (i > j) \\
B_{ij}^T &=& B_{ji} = A_{ji} (j>i) = A_{ij}^T P_R = P_R \circ A
\end{eqnarray*}
\end{proof}

\begin{lemma}\label{lemma_anti_symmetric}
Let $X \in \R^{N \times N}$ be an antisymmetric matrix, i.e. $X^T = - X$ and $P_L$ defined as above. We then can write
\begin{eqnarray}
X &=& P_L \circ X - (P_L \circ X)^T \;.
\end{eqnarray}
\end{lemma}
\begin{proof}
\begin{eqnarray*}
X &=& P_L \circ X + P_R \circ X = P_L \circ X + (P_L \circ X^T)^T = P_L \circ X - (P_L \circ X)^T
\end{eqnarray*}
\end{proof}

\begin{lemma}\label{lemma_abc}
Let $A,B,C \in \R^{M \times N}$. We then have
\begin{eqnarray}
\tr \left( A^T ( B \circ C)  \right) &=& \tr \left( C^T ( B \circ A) \right)
\end{eqnarray}
\end{lemma}
\begin{proof}
\begin{eqnarray*}
\tr ( A^T (B \circ C)) &=& \sum_{n=1}^N \sum_{m=1}^M A_{nm} B_{nm} C_{nm}= \tr ( C^T ( B \circ A ))
\end{eqnarray*}

\begin{lemma}\label{lemma_anti_diag}
Let $X \in \R^{N \times N}$ be antisymmetric, i.e. $X^T = -X$ and $D \in \R^{N \times N}$ be a diagonal matrix. Then we have
\begin{eqnarray}
0 &=& \Id \circ ( DX -XD)\;.
\end{eqnarray}
\end{lemma}
\begin{proof}
Define $A:= X D$ and $B := DX$, then the elements of $A$ and $B$ are given by $ A_{ik} = \sum_j X_{ij} D_{jk} \delta_{jk} = X_{ij} D_{kk}$ resp. $ B_{ik} = \sum_j D_{ij} \delta_{ij} X_{jk}  = D_{kk} X_{kk}$. Therefore the diagonal elements are $A_{kk}= B_{kk} = X_{kk} D_{kk}$.
\end{proof}

\section{Rectangular $QR$ decomposition}
We derive algorithms for the push forward (Algorithm \ref{push_forward_qr}) and pullback (Algorithm \ref{pullback_qr}) of the $QR$ decomposition $Q,R = \mathrm{qr}(A)$, where $A,Q \in R^{M \times N}$ and $R \in \R^{N \times N}$ for $M \geq N$.

\begin{algorithm}{Push forward of the Rectangular $QR$ decomposition:} \label{push_forward_qr} \\
$ A_d\in \R^{M\times N}$ , $Q_d \in \R^{M \times N}$ and $R_d \in \R^{N \times N}$, $ d=0,\dots,D-1$. We assume $M \geq N$, i.e. $A$ has more rows than columns.
\begin{itemize}
\item given: $\textblue{[A]_{D+E}}$, $1 \leq E \leq D$
\item compute  $\textred{[Q]_{D+E}} = \underbrace{\textred{[Q]_D}}_{\mbox{known}} + \underbrace{\textgreen{[\Delta Q]_E}}_{\mbox{wanted}} t^D$, \quad $\textred{[R]_{D+E}} = \underbrace{\textred{[R]_D}}_{\mbox{known}} + \underbrace{\textgreen{[\Delta R]_E}}_{\mbox{wanted}} t^D$
\begin{itemize}
 \item Step 1: \vspace*{-0.7cm}
\begin{eqnarray}
\; [\Delta F]_E t^D &\stackrel{D+E}{=}&   \textred{[Q]_D [R]_D} - \textblue{[A]_D} \\
\; [\Delta G]_E t^D &\stackrel{D+E}{=}&  \Id - \textred{[Q^T]_D [Q]_D}
\end{eqnarray}
\item Step 2:\vspace*{-0.7cm}
\begin{eqnarray}
\; [H]_E &\stackrel{E}{=}& \textblue{ [\Delta A]_E} - [\Delta F]_E \\
\; [S]_E &\stackrel{E}{=}& -\frac{1}{2} [\Delta G]_E
\end{eqnarray}
\item Step 3:\vspace*{-0.7cm}
\begin{eqnarray}
\; P_L \circ ([X]_E) &\stackrel{E}{=}& P_L \circ ( \textred{ [Q^T]_E} [H]_E \textred{ [R^{-1}]_E} ) - P_L \circ [S]_E \\
\end{eqnarray}
\item Step 4:\vspace*{-0.7cm}
\begin{eqnarray}
[K]_E &\stackrel{E}{=}& [S]_E + [X]_E
\end{eqnarray}
\item Step 5:\vspace*{-0.7cm}
\begin{eqnarray}
\; \textgreen{[\Delta R]_E} &\stackrel{E}{=}& \textred{[Q]_E^T} [H]_E - [K]_E \textred{[R]_E}
\end{eqnarray}
\item Step 6:\vspace*{-0.7cm}
\begin{eqnarray}
\; \textgreen{[\Delta Q]_E} &\stackrel{E}{=}& \left( [H]_E - [Q]_E [\Delta R]_E \right) [R]_E^{-1}
\end{eqnarray}
\end{itemize}
\end{itemize}
\end{algorithm}

\begin{proof}
The starting point is the implicit system
\begin{eqnarray*}
0 & \eqDE & \pfde{A} - \pfde{Q} \pfde{R} \\
0 & \eqDE & \pfde{Q^T} \pfde{Q} - \Id \\
0 & \eqDE & P_L \circ \pfde{R} \;.
\end{eqnarray*}
\begin{enumerate}
 \item To derive the algorithm, this system of equations has to be solved: \vspace*{-0.3cm}
\begin{eqnarray}
0 &\eqDE& \pfe{\Delta A} t^D +  \underbrace{\pfd{A} - \pfd{Q} \pfd{R}}_{ =: - \pfe{\Delta F} t^D}
 - \left(  \pfd{Q} \pfe{\Delta R} + \pfe{\Delta Q} \pfd{R} \right)  t^D \label{qr_1}\\
0 &\eqDE& \underbrace{ \pfd{Q^T} \pfd{Q} - \Id}_{=: - \pfe{\Delta G} t^D} + \left( \pfe{\Delta Q^T} \pfd{Q} + \pfd{Q^T} \pfe{\Delta Q} \right) t^D \label{qr_2}  \\
0 &\eqDE& P_L \circ \pfde{R} \label{qr_3}
\end{eqnarray}
\item Any matrix can be decomposed into a symmetric and an antisymmetric part: $\pfe{Q^T} \pfe{\Delta Q} \eqE \pfe{X} + \pfe{S}$ with $ \pfe{X}^T \eqE - \pfe{X}$ and $\pfe{S}^T \eqE \pfe{S}$. It then follows from \Eqn{qr_2} that $ 0 \eqE - \pfe{\Delta G} + 2 \pfe{S}$, i.e. we have $\pfe{S} \eqE \frac{1}{2} \pfe{\Delta G}$.
\item Insert the above relation into \Eqn{qr_1} and use \Eqn{qr_3} to obtain 
\begin{eqnarray*}
P_L \circ \pfe{X} &\eqE& P_L \circ \left( \pfe{Q^T} \underbrace{\left( \pfe{\Delta A} - \pfe{\Delta F}\right)}_{=:\pfe{H}}  \pfe{R}^{-1} \right) \;,
\end{eqnarray*}
uniquely defining $\pfe{X}$. Therefore, we have
\begin{eqnarray*}
\pfe{K} &:\eqE& \pfe{Q^T} \pfe{\Delta Q} \eqE \pfe{S} + \pfe{X} \;.
\end{eqnarray*}
\item to obtain $\pfe{\Delta R}$ we transform \ref{qr_1} and obtain
\begin{eqnarray*}
\pfe{\Delta R} &\eqE& \pfe{Q^T} \pfe{H} - \pfe{K} \pfe{R} \;.
\end{eqnarray*}
\item finally, transform \Eqn{qr_1} by right multiplication of $\pfe{R}^{-1}$:
\begin{eqnarray*}
\pfe{\Delta Q}  &\eqE& \left( \pfe{H} - \pfe{Q} \pfe{\Delta R} \right) \pfe{R}^{-1}
\end{eqnarray*}

\end{enumerate}
\end{proof}

The pullback transforms a function, here the $QR$ decomposition to a function that computes the adjoint function evaluation. Explicitely  $\pf{\alpha}( [\bar A], [A]) = \pf{\alpha}( [\bar Q], [Q]) + \pf{\alpha}( [\bar R], [R])$. 
\begin{algorithm}
  \label{pullback_qr}
The pullback can be written in a single equation
\begin{eqnarray}
\bar A &=& \bar A + Q \left( \bar R + P_L \circ \left(  R \bar R^T - \bar R R^T  + Q^T \bar Q - \bar Q^T Q \right) R^{-T} \right) + (\bar Q - Q Q^T \bar Q) R^{-T}\;.
\end{eqnarray}
For square $A$, i.e. $A \in R^{N \times N}$ the last term drops out. To lift the pullback one can use the same formula and use $X = [X]_D$ for $X$ either $A,Q,R,\bar A, \bar Q$ or $\bar R $.
\end{algorithm}

\end{proof}
\begin{proof}
We differentiate the implicit system
\begin{eqnarray*}
0 &=& A - Q R \\
0 &=& Q^T Q - \Id \\
0 &=& P_L \circ R
\end{eqnarray*}
and obtain
\begin{eqnarray*}
0 &=& \dd A - \dd Q R - Q \dd R  \quad (*)\\
0 &=& \dd Q^T Q + Q^T \dd Q \quad (**)\;.
\end{eqnarray*}
We define the antisymmetric ``matrix'' $ X := Q^T \dd Q$. Transforming equation $(*)$ as $Q^T (*) R^{-1}$ yields
\begin{eqnarray*}
0 &=& Q^T \dd A R^{-1} - Q^T Q \dd R R^{-1} - Q^T \dd Q \\
\mbox{therefore}\quad P_L \circ X &=& P_L \circ ( Q^T \dd A R^{-1}) \\
\mbox{and thus} \quad X &=& P_L \circ X - ( P_L \circ X)^T \;.
\end{eqnarray*}
Left multiplication $Q^T (*)$ yields $\dd R = Q^T \dd A - X R$ and transformation $(*)R^{-1}$ yields $\dd Q = ( \dd A - Q \dd R) R^{-1}$.
We are now in place to calculate
\begin{eqnarray*}
\tr (\bar Q^T \dd Q) + \tr ( \bar R^T \dd R) &=& \tr( \bar Q^T ( \dd A - Q \dd R) R^{-1} ) + \tr ( \bar R^T \dd R) \\
&=& \tr ( R^{-1} \bar Q^T \dd A) + \tr ( \underbrace{(\bar R^T - R^{-1} \bar Q^T Q)}_{=:F} \dd R) \\
&=&  \tr ( R^{-1} \bar Q^T \dd A) + \tr ( F ( Q^T \dd A - X R)) \\
&=& \tr (( R^{-1} \bar Q^T + F Q^T) \dd A) + \tr ( - R F X) \\
&=& \tr (( R^{-1} \bar Q^T + F Q^T) \dd A) + \tr ( - R F (P_L \circ X - (P_L \circ X)^T )) \\
&=& \tr (( R^{-1} \bar Q^T + F Q^T) \dd A) + \tr ( - R F (P_L \circ Q^T \dd A R^{-1} - (P_L \circ Q^T \dd A R^{-1})^T )) \\
&=& \tr (( \bar Q R^{-T} + Q F^T) \dd A^T) + \tr ( R^{-T} \dd A^T Q  (P_L \circ( RF - F^T R^T )) \\
&=& \tr ( (Q F^T + \bar Q R^{-T} + Q P_L ( R F - F^T R^T) R^{-T}) \dd A^T) \\
\mbox{therefore} \quad \bar A &=& \bar A + Q \left( \bar R + P_L \circ ( Q^T \bar Q - \bar Q^T Q + R \bar R^T - \bar R R^T) R^{-T} \right) + ( \bar Q - Q Q^T \bar Q) R^{-T} \;.
\end{eqnarray*}
In the above derivation we have used  Lemmas \ref{lemma_anti_symmetric}, \ref{lemma_transpose_projector} and \ref{lemma_abc}.

\end{proof}

\section{Eigenvalue Decomposition of Symmetric Matrices with Distinct Eigenvalues}
We compute for $A \in R^{N \times N}$ a symmetric matrix with distinct eigenvalues the eigenvalue decomposition $A Q = Q \Lambda$, where $\Lambda \in \R^{N \times N}$ is a diagonal matrix and $Q \in \R^{N \times N}$ an orthonormal matrix.

\begin{algorithm}{Push Forward of the Symmetric Eigenvalue Decomposition with Distinct Eigenvalues:}\\
We assume  $\textred{[Q]_{D+E}} = \textred{[Q]_D} +\textgreen{[\Delta Q]_E} t^D$ and $\textred{[\Lambda]_{D+E}} = \textred{[\Lambda]_D} +\textgreen{[\Delta \Lambda]_E} t^D$, where $\QD$ and $\LD$ have already been computed. I.e., it is the goal to compute the next $E$ coefficients  $\DQE$ and $\DLE$.
\begin{itemize}
 \item Step 1:
 \vspace*{-0.7cm}
\begin{eqnarray}
\; \DFE t^D &\stackrel{D+E}{=}&  \QTD \AD \QD - \LD  \\
\; \DGE t^D &\stackrel{D+E}{=}&  \QTD \QD - \Id  
\end{eqnarray}
 \item Step 2:
 \vspace*{-0.7cm}
\begin{eqnarray}
\; [S]_E  &\stackrel{E}{=}&  - \frac{1}{2} \DGE 
\end{eqnarray}
 \item Step 3:
 \vspace*{-0.7cm}
\begin{eqnarray}
\; [K]_E  &\stackrel{E}{=}&  \DFE + \QTE \DAE \QE + [S]_E \LE + \LE [S]_E
\end{eqnarray}
 \item Step 4:
 \vspace*{-0.7cm}
\begin{eqnarray}
\; \DLE  &\stackrel{E}{=}&  \Id \circ [K]_E
\end{eqnarray}
 \item Step 5:
 \vspace*{-0.7cm}
\begin{eqnarray}
\; [H_{ij}]_E &\stackrel{E}{=}& ([\lambda_j]_E - [\lambda_i]_E)^{-1} \quad \mbox{if} \quad  i \neq j , \quad 0 \quad \mbox{ else}
\end{eqnarray}
 \item Step 6:
 \vspace*{-0.7cm}
\begin{eqnarray}
\; \DQE &\stackrel{E}{=}& \QE \left( [H]_E \circ ( [K]_E - \DLE) + [S]_E \right)
\end{eqnarray}
\end{itemize}
\end{algorithm}
\begin{proof}
We solve the implicit system
\begin{eqnarray*}
0 &\eqDE& \QTDE \ADE \QDE  - \LDE  \\ 
0 &\eqDE& \QTDE \QDE - \Id \\
0 &\eqDE& (P_L + P_R) \circ \LDE \;.
\end{eqnarray*}
We split $\pfde{A} = \pfd{A} + \pfe{\Delta A}t^D$ in the known part $\pfd{A}$ and unknown part $\pfe{\Delta A}$. For ease of notation we use $A \equiv \pfd{A}$  and $\Delta A \equiv \pfe{\Delta A}$ in the following.
The first equation can be transformed as follows:
\begin{eqnarray*}
0 & \eqDE & \pfde{Q^T} \pfde{A} \pfde{Q}  - \pfde{\Lambda} \\
 &  \eqDE & (Q^T + \Delta Q^T t^D) (A + \Delta A t^D)   (Q + \Delta Q t^D) - ( \Lambda + \Delta \Lambda t^D) \\
 &  \eqDE & \underbrace{Q^T A Q - \Lambda}_{ =: \Delta F t^D} + \left( Q^T A \Delta Q + Q^T \Delta A Q + \Delta Q^T A Q - \Delta \Lambda \right) t^D
\end{eqnarray*}

From the second equation we get 
\begin{eqnarray*}
\Delta G & \eqE & Q^T \Delta Q + \Delta Q^T Q \;,
\end{eqnarray*}
which can be written as 
\begin{eqnarray*}
S + X &\eqE & \Delta Q^T Q \;,
\end{eqnarray*}
where $S$ is symmetric and $X$ is antisymmetric.

Using $ A Q - Q \Lambda$ we get
\begin{eqnarray*}
0 & \eqE & \Delta F - \Delta \Lambda  + Q^T \Delta A Q + \Lambda Q^T \Delta Q  + \Delta Q^T Q \Lambda \\
\Delta \Lambda  & \eqE & \Delta F  + Q^T \Delta A Q + \Lambda ( S - X)  + ( S + X) \Lambda \\
\end{eqnarray*}
Using the structural information that $\Lambda$ is a diagonal matrix we obtain
\begin{eqnarray*}
\Delta \Lambda  & \eqE & \Id \circ \left( \Delta F  + Q^T \Delta A Q + \Lambda ( S - X)  + ( S + X) \Lambda \right)\\
& \eqE & \Id \circ \underbrace{\left( \Delta F  + Q^T \Delta A Q + \Lambda S  + S \Lambda \right)}_{=:K} \;,
\end{eqnarray*}
since $ \Id \circ (X \Lambda - \Lambda X) = 0$ from Lemma \ref{lemma_anti_diag}. Now that we have $\Delta \Lambda$ we can compute $\Delta Q$:
\begin{eqnarray*}
0 & \eqE &  K - \Delta \Lambda + X \Lambda -\Lambda X \\
0 & \eqE &  K - \Delta \Lambda + E \circ X \\
\end{eqnarray*}
where $E_{ij} = \Lambda_{jj} - \Lambda_{ii}$. Define $H_{ij} = ( \Lambda_{jj} - \Lambda_{ii} )^{-1}$ for $i\neq j$ and $0$ else:
\begin{eqnarray*}
X^T & \eqE & H \circ \left( K - \Delta \Lambda \right) \;.
\end{eqnarray*}
Using $X^T \eqE Q^T \Delta Q - S$ we obtain
\begin{eqnarray*}
\Delta Q & \eqE &  Q \left( H \circ \left( K - \Delta \Lambda \right) + S  \right) \;.
\end{eqnarray*}
This concludes the proof.
\end{proof}

We now derive the pullback formulas for the eigenvalue decomposition. The pullback acts as $\overleftarrow{P}: (A \mapsto Q, \Lambda) \mapsto  ((\bar Q, \Lambda , Q, \Lambda, A) \mapsto \bar A)$

\begin{algorithm}{Pullback of the Symmetric Eigenvalue Decomposition with Distinct Eigenvalues:}\\
Given $\textred{A}, \textred{Q}, \textred{\Lambda}, \textred{\bar Q}, \textred{\bar \Lambda}$, compute $\textgreen{\bar A}$
\begin{eqnarray}
\; [H_{ij}]_D &\stackrel{D}{=}& (\textred{[\lambda_j]_D} - \textred{[\lambda_i]_D})^{-1} \quad \mbox{if} \quad  i \neq j , \quad 0 \quad \mbox{ else} \\
\textgreen{[\bar A]_D} &\stackrel{D}{=}& \textred{[Q]_D} \left( \textred{[\bar \Lambda]_D} + [H]_D \circ (\textred{[Q^T]_D [\bar Q]_D})\right) \textred{[Q^T]_D}
\end{eqnarray}
\end{algorithm}

\begin{proof}
We want to compute $ \tr( \bar A^T \dd A) = \tr( \bar \Lambda^T \dd \Lambda) + \tr( \bar Q^T \dd Q)$. We differentiate the implicit system
\begin{eqnarray*}
0 &=& Q^T A Q - \Lambda \\
0 &=& Q^T Q - \Id\\
0 &=& (P_L + P_R) \circ \Lambda
\end{eqnarray*}
and obtain
\begin{eqnarray*}
\dd \Lambda &=& \dd Q^T A Q + Q^T \dd A Q + Q^T A \dd Q \\
0 &=& \dd Q^T Q + Q^T \dd Q \;.
\end{eqnarray*}
A straight forward calculation shows:
\begin{eqnarray*}
\tr ( \bar \Lambda^T \dd \Lambda) &=& \tr( A Q \bar \Lambda^T \dd Q^T) + \tr( \Lambda Q^T A \dd Q) + \tr ( Q \bar \Lambda Q^T \dd A) \\
&=& \tr ( Q^T A Q \bar \Lambda^T \dd Q^T Q) + \tr ( \bar \Lambda Q^T A Q Q^T \dd Q)  + \tr ( Q \bar \Lambda Q^T \dd A) \\
&=& \tr ( \Lambda \bar \Lambda \dd Q^T Q ) + \tr( \bar \Lambda \Lambda Q^T \dd Q) + \tr ( Q \bar \Lambda Q^T \dd A) \\
&=& \tr ( Q \bar \Lambda Q^T \dd A) \;, \\
\tr( \bar Q^T \dd Q) &=& ( \bar Q^T Q Q^T \dd Q) \\
&=& \tr (\bar Q^T Q (H \circ (Q^T \dd A Q))) \\
&=& \tr ( Q^T \dd A^T Q ( H \circ (Q^T \bar Q))) \\
&=& \tr ( Q ( H^T \circ (\bar Q^T Q)) Q^T \dd A) \;, \\
\tr ( \bar A^T \dd  A) &=& \tr \left( ( Q (\bar \Lambda + H^T \circ ( \bar Q^T Q)) Q^T )\dd A \right)
\end{eqnarray*}
where we have used
\begin{eqnarray*}
0 &=& AQ - Q \Lambda \\
\Rightarrow 0 &=& \dd A Q + A \dd Q - \dd Q \Lambda - Q \dd \Lambda \\
&=& \dd A Q + Q Q^T A Q Q^T \dd Q - Q Q^T \dd Q \Lambda - Q \dd \Lambda \\
&=& \dd A Q - Q ( K \circ (Q^T \dd Q)) - Q \dd \Lambda \\
Q^T \dd Q &=& H \circ ( Q^T \dd A Q - \dd \Lambda) \\
&=& H \circ ( Q^T \dd A Q)
\end{eqnarray*}
where we have defined $K_{ij} := \Lambda_{jj} - \Lambda_{ii}$ and $H_{ij} = (K_{ij})^{-1}$ for $i \neq j$ and $H_{ij} = 0$ otherwise and used the property $ \Lambda X - X \Lambda = K \circ X$ with $K_{ij} = \Lambda_{jj} - \Lambda_{ii}$ for all $X \in \R^{N \times N}$ and diagonal $\Lambda \in \R^{N \times N}$.
\end{proof}

\section{Software Implementation}
The algorithms presented here are implemented in the Python AD tool \texttt{ALGOPY} \cite{ALGOPY}. It is BSD licensed and can be publicly accessed at \texttt{www.github.com/b45ch1/algopy}. The goal is to provide code that is useful not only for end users but also serving as repository of tested algorithms that can be easily ported to other programming languages or incorporated in other AD tools. The focus is on the implementation of polynomial factor rings and not so much on the efficient implementation of the global derivative accumulation. The global derivative accumulation is implemented by use of a code tracer like ADOL-C \cite{Griewank1999ACA} but stores the computational procedure not on a sequential tape but in computational graph where each function node also knows it's parents, i.e. the directed acyclic graph is stored in a doubly linked list.

 At the moment, prototypes for univariate Taylor propagation of scalars, cross Taylor propagation of scalars and univariate Taylor propagation of matrices are implemented. It is still in a pre-alpha stage and the API is very likely to change.

\section{Preliminary Runtime Comparison}
This section gives a rough overview of the runtime behavior of the algorithms derived in this paper compared to the alternative approach of differentiating the linear algebra routines. We have implemented a $QR$ decomposition algorithm that can be traced with \texttt{PYADOLC}. The results are shown and interpreted in Figure \ref{fig:algopy_vs_pyadolc}. In another test we measure the ratio between the push forward runtime and the normal function evaluation runtime:
\begin{eqnarray*}
\mbox{TIME(push forward)}/\mbox{TIME(normal)} \approx  11.79
\end{eqnarray*}
 for the $QR$ decomposition for $A \in \R^{100 \times 5}$ up to degree $D=4$ and five parallel evaluations at once.
For the eigenvalue decomposition we obtain 
\begin{eqnarray*}
\mbox{TIME(push forward)}/\mbox{TIME(normal)} \approx  11.88
\end{eqnarray*}
for $A \in \R^{20 \times 20}$, $D=4$ and five parallel evaluations. This test is part of \texttt{ALGOPY} \cite{ALGOPY}.

\begin{figure}
\centering
\includegraphics[width=0.49\textwidth]{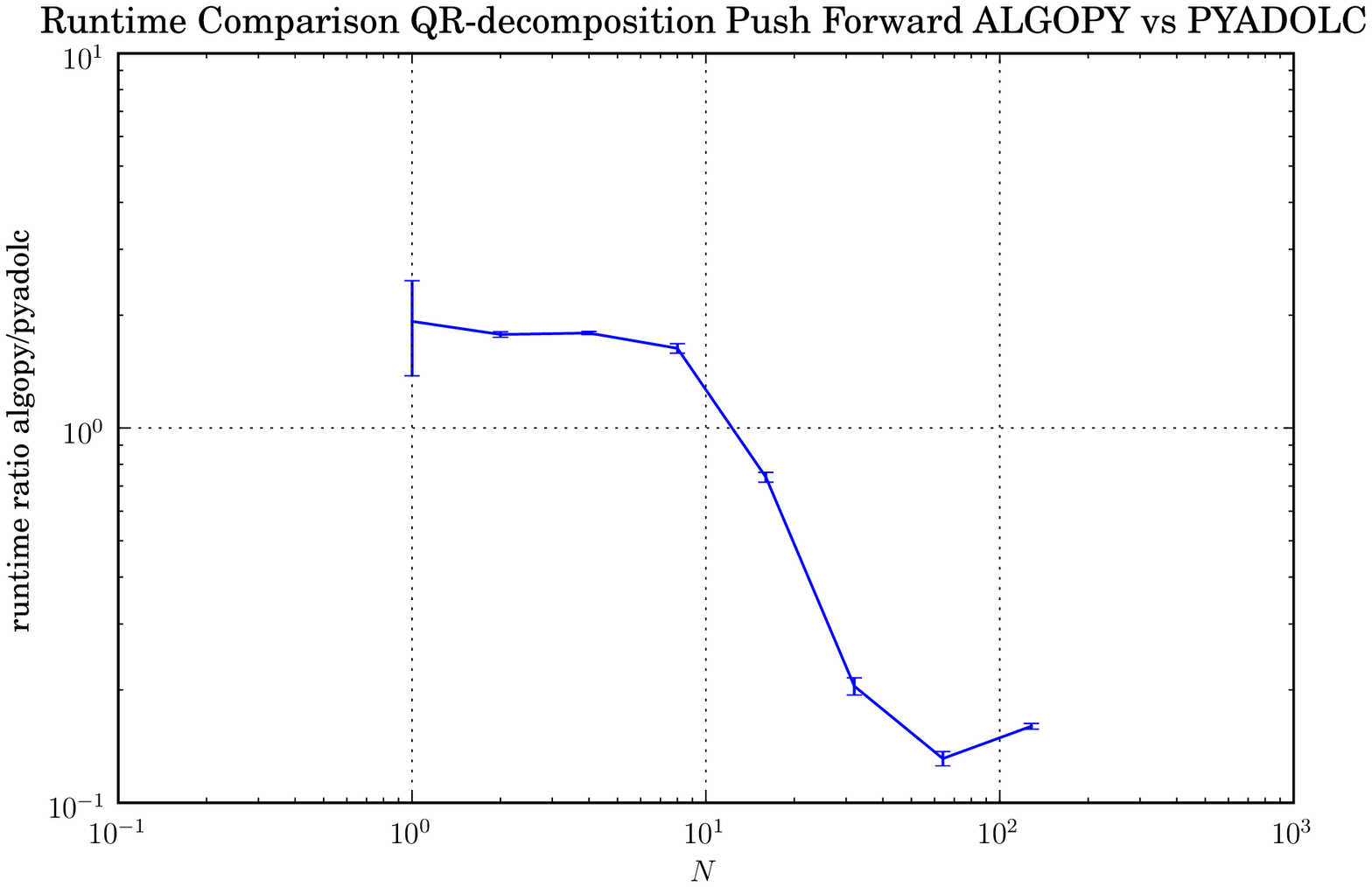}
\includegraphics[width=0.49\textwidth]{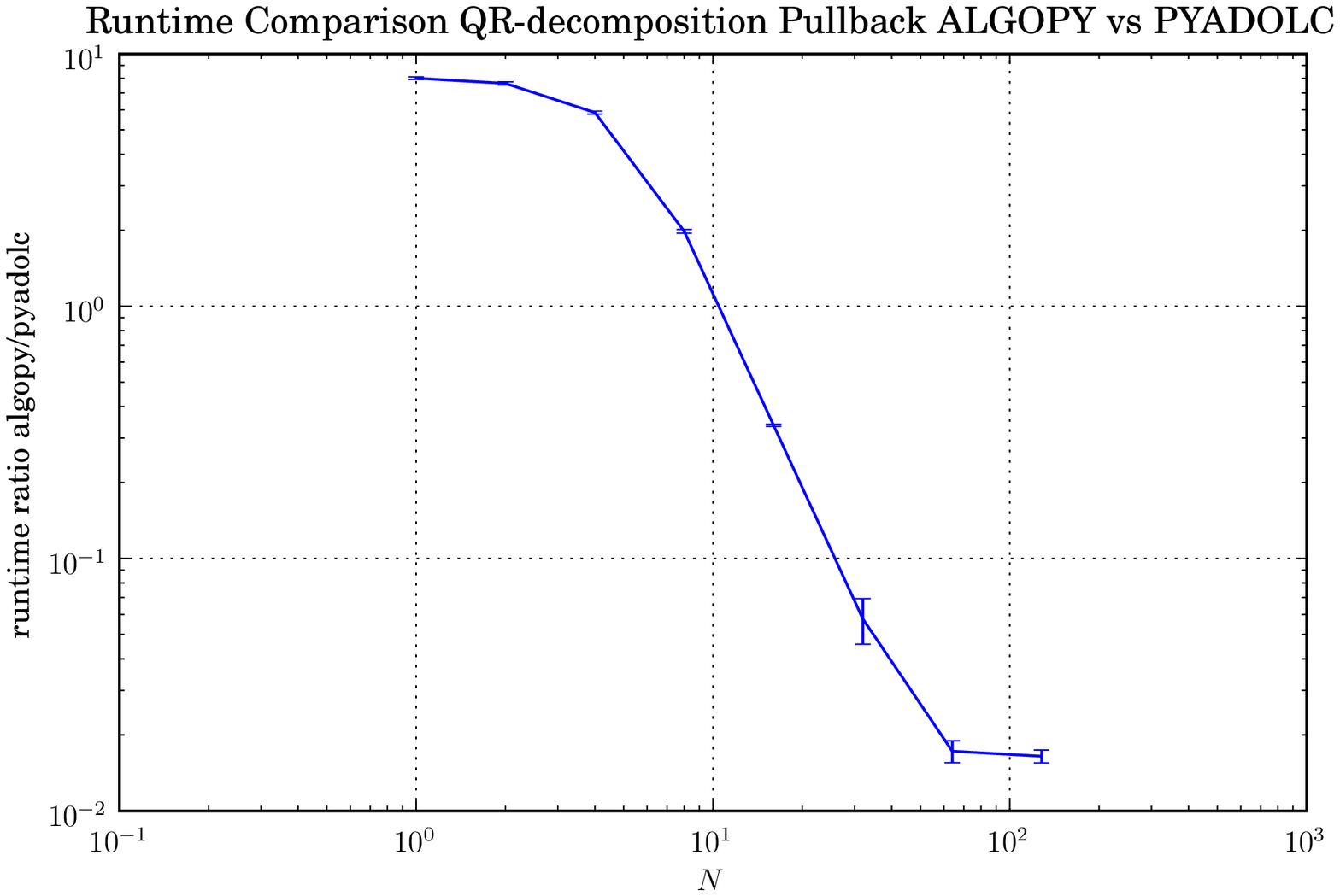}
\caption{\label{fig:algopy_vs_pyadolc}
In the left plot the comparison between the push forward of \texttt{PYADOLC} and \texttt{ALGOPY} for $D=4$. In the right plot the lifted pullback for $D=4$. On the $x$-axis we plot the size $N$ of the test matrices $A \in \R^{N \times N}$ and on the $y$-axis the runtime ratio. Unfortunately, there are significant fluctuation in the relative runtime measurements. Therefore, we have repeated each test 10 times and plotted mean and standard deviation.  Nonetheless it is obvious that for large matrices the UTPM implementation in \texttt{ALGOPY} clearly outperforms the UTPS differentiated algorithm using \texttt{PYADOLC}. We must stress that the plots only indicate the actual runtime ratio that would be obtained by efficient FORTRAN/C/C++ implementation of both methods. There are also many possibilities to improve the performance of \texttt{PYADOLC}, e.g. by adjusting the buffer sizes of ADOL-C or by using direct LAPACK calls instead of the standard \texttt{numpy.dot} function in \texttt{ALGOPY}.
 }
\end{figure}

\section{Example Program: Gradient Evaluation of an Optimum Experimental Design objective function}
The purpose of this section is to show a motivating example from optimum experimental design where these algorithms are necessary to compute the gradient of the objective function $\Phi$ in a numerically stable way.
The algorithmic procedure to compute $\Phi$ is given as a straight-line program
\begin{eqnarray*}
F &=& F(x,y) \\
J &=& \frac{\dd F}{\dd y} \\
Q,R &=& \qr(J) \\
D &=& \solve(R,\Id) \\
C &=&  D D^T \\
\Lambda, U &=& \eig(C) \\
\Phi &=& \Lambda_{11} \\
\end{eqnarray*}
where $F \in \R^{N_m}$, $\gamma \in \R$, $J \in \R^{N_m \times N_y}$.
The $QR$ decomposition is used for numerical stability reasons since otherwise the multiplication of $J^T J$ would square the condition number.

To be able to check the correctness of the computed gradient we use a simple $F(x,y)$ that allows us to derive an analytical solution by symbolic differentiation. We use $F(x,y) = B x y $ where $B \in \R^{N_m \times N_x}$ is a randomly initialized matrix, $x \in \R^{N_x}$ and $y \in \R$. Thus, the objective function is $\Phi(y) = y^{-2} \lambda_1( (B^T B)^{-1}$ and thus $\nabla_y \Phi(y) = -2 y^{-3} \lambda_1( (B^T B)^{-1}$. We use $N_x = 11$. A typical test run where the symbolical solution and the AD solution are compared yields
\begin{eqnarray*}
 |(\nabla_y \Phi)^{\mbox{symbolic}} - (\nabla_y \Phi)^{\mbox{AD}}  | &=& 4.4 \times 10^{-15} \;.
\end{eqnarray*}
This example is part of \texttt{ALGOPY} \cite{ALGOPY}.